\documentclass{llncs} 
\newtheorem{defn}{Definition} 

\newtheorem{thm}{Theorem}
\newtheorem{lem}[defn]{Lemma}
\newtheorem{cor}{Corollary}
\newtheorem {rem}{Remark}
\newtheorem{con}{Conjecture}

\title{Row monomial matrices and Cerny conjecture, short proof}

\date{2.4.2022}
\author{A.N. Trahtman\thanks{Email: avraham.trakhtman@gmail.com}
\institute{{Bar-Ilan University, Dep. of Math., 52900, Ramat Gan, Israel}}
}

\begin{document}

\maketitle

 \begin{abstract}

  A word $w$ of letters on edges of underlying graph $\Gamma$ of  deterministic
 finite automaton (DFA) is called synchronizing  if $w$ sends all states of
the automaton to a unique state.  J. \v{C}erny discovered in 1964 a sequence of
 $n$-state complete DFA possessing a minimal synchronizing word of length $(n-1)^2$.

The hypothesis, well known today as the \v{C}erny conjecture, claims that  $(n-1)^2$
is also precise  upper bound on the length of synchronizing word for a complete DFA.
The hypothesis was formulated in 1966 by Starke.
The problem has motivated great and constantly growing number of investigations
 and generalizations.

We present the proof  of the \v{C}erny-Starke conjecture: the deterministic complete
$n$-state synchronizing automaton has synchronizing word of length at most $(n-1)^2$.

The  proof used also the space of row monomial matrices (one unit and rest zeros in every row) 
and connection between dimension of the space and the length of words on paths of edges in
underlying graph of automaton.

\end{abstract}

\section*{Introduction}

The long and fascinating history of state machine synchronization and the problems around
was reflected in hundreds articles.
 The problem of synchronization of finite automata is a natural one and various aspects
of this problem have been touched in the literature.
Different problems of synchronization and achievements one can find in surveys 
\cite{Ju}, \cite{KV} and works \cite{B}, \cite{MS}, \cite{Tm}.

The synchronizing word limits the propagation of errors for a prefix code.
Deterministic finite automaton is a tool that helps to recognized language in a set of DNA strings.

A problem with a long story is the estimation of the minimal length of synchronizing word.
 J. \v{C}erny in 1964 \cite{Ce} found the infinite sequence of $n$-state complete DFA
with shortest synchronizing word of length $(n-1)^2$ for an alphabet of size two.
Since then, only 27 small automata with synchronizing word of length $(n-1)^2$
of size $n \leq 6$ have been added to the single \v{C}erny sequence \cite{Tm}.
  The following hypothesis is well known today as \v{C}erny's conjecture:

\begin{con}
The deterministic complete $n$-state synchronizing automaton over alphabet $\Sigma$
has synchronizing word in $\Sigma$ of length at most $(n-1)^2$ \cite{Sta} (Starke, 1966).
  \end{con}

The problem can be reduced to automata with a strongly connected graph \cite{Ce}.

We consider a class of matrices $M_u$ of mapping induced by words $u$
in the alphabet of letters on edges of the underlying graph $\Gamma$. Every edge 
has length one.

A sequence of distinct non-trivial spaces ordered by inclusion of matrices
is used in the study.

Initially found upper bound for the minimal length of synchronizing word was big and
has been consistently improved over the years by different authors.
The upper bound found by Frankl in 1982 \cite{Fr} is equal to $(n^3-n)/6$.
The result was reformulated in terms of synchronization in \cite{Pin}
and repeated independently in \cite{KRS}.
The cubic estimation of the bound exists since 1982.

Examples of automata such that the length of the shortest synchronizing word is greater 
than $(n-1)^2$ are unknown. Moreover, the examples of automata  with shortest 
synchronizing word of length $(n-1)^2$ are infrequent. 
After the sequence of \v{C}erny and the example of \v{C}erny, Piricka and Rosenauerova \cite{CPR} of 1971 for $|\Sigma|=2$, the next such examples were found by
 Kari \cite {Ka} in 2001 for $n=6$ and $|\Sigma|=2$ and by Roman \cite {Ro}
 for $n=5$ and $|\Sigma|=3$ in 2004.  The package TESTAS \cite {TS} studied 
all automata with strongly connected underlying graph of restricted size and alphabet from $2$ to $4$ found five new examples of DFA with shortest synchronizing 
word of length $(n-1)^2$ for $n\leq 4$.

Don and Zantema present in \cite {DZ} an ingenious method of designing several new 
automata, a kind of "hybrids" from existing examples  from \cite{Ce}, \cite{CPR}, \cite {TS} 
of size three and four. They proved that for $n\geq 5$ the method does not work.
So there are up to isomorphism exactly 15 DFA for $n=3$ and 12 DFA for $n=4$ 
with shortest synchronizing word of length $(n-1)^2$.

The authors of \cite {DZ} support the hypothesis from \cite{TS} that all automata 
with shortest synchronizing word  of length $(n-1)^2$ are known, of course,
 with essential correction found by themselves for $n=3,4$.

We propose below an attempt to prove the \v{C}erny conjecture.

\section*{Preliminaries}
We consider a complete $n$-state DFA with
 strongly connected underlying graph $\Gamma$
 over a fixed finite alphabet $\Sigma$ of labels on edges of  $\Gamma$ of an automaton $A$.
The trivial cases $n \leq 2$, $|\Sigma|=1$ and $|A \sigma|=1$ for
$\sigma \in \Sigma$ are excluded.

The restriction on strongly connected graphs is based on \cite{Ce}.
The states of the automaton $A$ are considered also as vertices of the graph $\Gamma$.

If there exists a path in an automaton from the state $\bf p$ to
the state $\bf q$ and the edges of the path are consecutively
labelled by $\sigma_1, ..., \sigma_k$, then for
$s=\sigma_1...\sigma_k \in \Sigma^+$ let us write ${\bf q}={\bf p}s$.

Let $Px$ be the set of states ${\bf q}={\bf p}x$ for all ${\bf p}$
from the subset $P$ of states and $x \in \Sigma^+$.
Let $Ax$ denote the set $Px$ for the set $P$ of all states of the automaton.

 A word $s \in \Sigma^+ $ is called {\it synchronizing (reset, magic, recurrent, homing)}
 word of an automaton $A$ with underlying graph $\Gamma$ if $|As|=1$.
The word $s$ below denotes minimal synchronizing word.

We consider row monomial matrices (one unit and rest of zeros in every row)
$n \times n$-matrices and edges on underlying graph of length one.

The states of the automaton are enumerated, the state $\bf q$  has number one.

 An automaton (and its underlying graph) possessing a synchronizing word is called {\it synchronizing}.

We connect a mapping of the set of states of the automaton made by a word $u$ of
row monomial $n \times n$-matrix $M_u$ such that for an element $m_{i,j} \in M_u$ takes place

\centerline{$m_{i,j}$= $\cases{1, &${\bf q}_i u ={\bf q}_j$; \cr 0, &otherwise.}$}

Any mapping of the set of states of the automaton  $A$  can be presented by some  word $u$
and by a corresponding matrix $M_u$.
For instance,

 \centerline{$M_u = \left(
\begin{array}{ccccccc}
  0 & 0 & 1 & . & . & . &  0 \\
  1 & 0 & 0 & . & . & . &  0 \\
  0 & 0 & 0 & . & . & . &  1 \\
  1. & . & . & . & . & . &  . \\
  0 & 1 & 0 & . & . & . &  0 \\
  1 & 0 & 0 & . & . & . &  0 \\
\end{array}\right)
$}

 Let us call the matrix $M_u$ of the mapping induced by the word $u$, for brevity,
the matrix of word $u$.

$M_uM_v=M_{uv}$ \cite{Be}.

The set of nonzero columns of $M_u$ (set of second indexes of its elements) of $M_u$
 is denoted as $R(u)$ of size $|R(u)|$.

For linear algebra terminology and definitions, see \cite{Ln}, \cite{Ma}.

\section{Some properties of row monomial matrices}

\begin{rem} \label{r1}
The invertible matrix $M_a$ does not change the number of units of every column of $M_u$ in its image of the product $M_aM_u$.

Every unit in the product $M_uM_a$ is the product of two units, first unit from nonzero column of $M_u$
and second unit from a row with one unit of $M_a$.

\end{rem}

\begin{rem} \label{r4}

The columns of the matrix $M_uM_a$ are obtained by permutation of columns $M_u$.
Some columns can be merged (units of columns are moved along
row to a common column) with $|R(ua)|<|R(u)|$.

The rows of the matrix $M_aM_u$ are obtained by permutation of rows of the matrix $M_u$.
Some of these rows may disappear and replaced by another rows of $M_u$.

\end{rem}

\begin{lem} \label{l1}

The number of nonzero columns $|R(b)|$ is equal to the rank of $M_b$.

\centerline{$|R(ua)| \leq |R(u)|$ and $R(au) \subseteq R(u)$.}

For invertible matrix $M_a$ we have $R(au)=R(u)$ and $|R(ua)|=|R(u)|$.

 Nonzero columns of $M_{ua}$ have units also in $M_a$.

\end{lem}

\begin{proof}
The matrix $M_b$ has submatrix with nonzero determinant having only one unit in every
row and in every nonzero column.
Therefore $|R(b)|$ is equal to the rank of $M_b$.

The matrix $M_a$ in the product $M_uM_a$ shifts column of
$M_u$ to columns of $M_uM_a$ without
changing the column itself by Remark \ref{r4} or merging
some columns of $M_u$.
In view of possible merged columns, $|R(ua)|\leq |R(u)|$.

Some rows of $M_u$ can be replaced in $M_aM_u$ by another row and therefore some rows
from $M_u$ may be changed, but zero columns of $M_u$ remain in $M_aM_u$ (Remark 1).

Hence $R(au) \subseteq R(u)$ and $|R(ua)| \leq |R(u)|$.

For invertible matrix $M_a$ we have $R(au)= R(u)$  and $|R(ua)|=|R(u)|$.

Nonzero columns of $M_{ua}$ have units also in $M_a$ in view of $R(ua) \subseteq R(a)$.

\end{proof}

\begin{cor}  \label{c1}
All matrices of prefixes of synchronizing matrix $s$ also have
at least one unit in nonzero column of $s$.
\end{cor}

\begin{cor}  \label{c3}
The invertible matrix $M_a$ keeps the number of units of any column of $M_u$
 in corresponding column of the product $M_aM_u$.
\end{cor}

$ |R(u)|$ we call also rank of word $u$.

\subsection{Dimension of the space of row monomial matrices.}

\begin{lem}  \label {v3}

 The set $V$ of all $n\times k$-matrices having
precisely one unit in any row and zeroes everywhere else (row monomial)
(or $n\times n$-matrices with zeros in fixed $n-k$ columns for $k\leq n$) has at most
$n(k-1)+1$ linear independent matrices.
 \end{lem}

\begin{proof}
Let us consider distinct $n\times k$-matrices of word with at most only one nonzero 
cell outside the last nonzero column $k$.

Let us begin from the matrices $V_{i,j}$ with unit in $(i,j)$ cell ($j<k$) and units in ($m,k$) 
cells for all $m$ except $i$.
The remaining cells contain zeros.
So we have $n-1$ units in the $k$-th column and only one unit in remaining $k-1$ columns 
of the matrix $V_{i,j}$.
Let the matrix $K$ have units in the $k$-th column and zeros in the other columns.
There are $n(k-1)$ matrices $V_{i,j}$. Together with $K$ they belong to the set $V$.
So we have $n(k-1)+1$ matrices. For instance,

\begin{picture}(0,40)
\end{picture}
$V_{1,1}=\left(
\begin{array}{cccccccc}
  1 & 0 & 0 & . & . & 0  \\
  0 & 0 & 0 & . & . & 1  \\
  0 & 0 & 0 & . & . & 1  \\
  . & . & . & . & . & .  \\
  0 & 0 & 0 & . & . & 1  \\
  0 & 0 & 0 & . & . & 1  \\
\end{array}
\right)$
\begin{picture}(4,40)
\end{picture}
$V_{3,2}=\left(
\begin{array}{cccccccc}
  0 & 0 & 0 & . & . & 1  \\
  0 & 0 & 0 & . & . & 1  \\
  0 & 1 & 0 & . & . & 0  \\
  . & . & . & . & . & .  \\
  0 & 0 & 0 & . & . & 1  \\
  0 & 0 & 0 & . & . & 1  \\
\end{array}
\right)$
\begin{picture}(4,40)
\end{picture}
$K=\left(
\begin{array}{cccccccc}
  0 & 0 & 0 & . & . & 1 \\
  0 & 0 & 0 & . & . & 1 \\
  0 & 0 & 0 & . & . & 1 \\
  . & . & . & . & . & . \\
  0 & 0 & 0 & . & . & 1 \\
  0 & 0 & 0 & . & . & 1 \\
\end{array}
\right)$

 The first step is to prove that the matrices $V_{i,j}$ and $K$ generate the space with the set $V$.
For arbitrary matrix $T$ of word from $V$ for every $t_{i,j} \neq 0$ and $j<k$,
let us consider the matrices $V_{i,j}$ with unit in the cell $(i,j)$ and the sum of them $Z=\sum V_{i,j}$.
So the first $k-1$ columns of $T$ and $Z$ coincide.

   Hence in the first $k-1$ columns of the matrix $Z$ there is at most only one unit in any row.
 Therefore in the cell of $k$-th column of $Z$ one can find only value of $m$ or $m-1$.
The value of $m$ appears if there are only zeros in other cells of the considered row. 
Therefore $\sum V_{i,j} - (m-1)K=T$.
Thus every matrix from the set $V$ is a span of $(k-1)n +1$ matrices from $V$.

It remains now to prove that the set of matrices $V_{i,j}$ and $K$ is a set of linear independent matrices.
If one excludes a certain matrix $V_{i,j}$ from the set of these matrices, then it is impossible
 to obtain a nonzero value in the cell $(i,j)$ and therefore to obtain the matrix $V_{i,j}$.
So the set of matrices $V_{i,j}$ is linear independent.

Every non-trivial linear combination of the matrices $V_{i,j}$ equal to a matrix of word has at
 least one nonzero element in the first $k-1$ columns.
Therefore, the matrix $K$ could not be obtained as a linear combination of the matrices $V_{i,j}$.
Consequently the set of matrices $V_{i,j}$ and $K$ forms a basis of the set $V$.
\end{proof}

\begin{cor}  \label {c2}
The set of all row monomial $n \times(n-1)$-matrices of words has at most $(n-1)^2$
 linear independent matrices.

The set of row monomial  $n\times n$-matrices has at most $n(n-1)+1$ linear independent matrices.

There are at most $n+1$ linear independent row monomial matrices of
words in the set of matrices with 2 nonzero columns.
and at most $n$  row monomial linear independent matrices  in the set of matrices
with one common nonzero column.

 \end{cor}

\begin{cor}  \label {cs}

There exists a sequence  of length at most  $n(n-1)+1$ of distinct not empty subspaces ordered by inclusion
of  matrices for $n$-state automaton.

\end{cor}

\section{Linear independent matrices $M_u$}

\begin{lem} \label{v11}

Let the space $W$ be generated by linear independent 
$n\times n-$matrices $M_u$ of words $u$ with  $|R(u)|>1$ of
synchronizing complete strongly connected DFA.

Then some matrix $M_{u\beta} \not\in W$ for generator $M_u$ of $W$
and some letter $\beta$.

\end{lem}

 \begin{proof}

Assume the contrary: for every generator $M_{u_i}$ of $W$
and every letter $\beta$ the matrix $M_{u_i\beta}$ in $W$.

So  the space $W$ has the same set of generators.
Consequently for every generator $M_{u_i}$ and every word $t$ the product
$M_{u_it} \in W$.

For every matrix  $M_v \in W$

\centerline{$M_v=\sum \tau_i M_{u_i}$}
with generators $M_{u_i}$ in  $W$ . Then

\centerline{$M_vM_{\beta}=(\sum \tau_i M_{u_i})M_{\beta}=\sum \tau_i M_{u_i}M_{\beta}$}
for matrices $M_{u_i}M_{\beta}= M_{u_i\beta}$ un $W$.

Therefore also $M_{v\beta}=M_vM_{\beta}=M_v \sum \tau_iM_{u_i\beta}$ belongs to $W$.

By induction for every word $t$ the matrix $M_{u_i t}$ belongs to $W$
and  the matrix $M_{vt} \in W$.

However in synchronizing complete  strongly connected DFA for some generator
$M_{u_i}$ and some $t$  $|R(u_it)|=1$ outside $W$.

\end{proof}

\begin{cor} \label{v12}

Let the sequence of spaces $W_j$ of  matrices be ordered  by inclusion
with increasing  $j$. The basis of $W_0$ contains the matrix $M_0$ with $|R(u)|=n$
($n$ units in one row, say, $q$).

By Lemma \ref{v11}, for  the space $W_{j+1}$  $W_{j+1} \supset  W_j$.
The space $W_j$ is can be extended by matrix $M_{u\beta}$ of a letter $\beta$ 
and some  $M_u$ from the basis of $W_j$ by Lemma \ref{v11}.

Matrices from basis $W_j$ are included in basis  $W_{j+1}$ as far as possible.

Hence by induction in the sequence of all  matrices of kind $M_{u\beta}$ all matrices 
are linear independent.

 \end{cor}

\section{The sequence of spaces $W_j$ ordered by inclusion}

We study matrices from $W_j$. Moreover, $W_{j+1}$ includes basis  $W_j$. 
 Matrices of generators of $W_j$ are linear independent (Corollary \ref{v12}).
All generators were obtained by adding a letter from right to
 generator of the former  $W_{j-1}$.

The rank $|R(u)|$ plays some role in the study.

By Corollary \ref{cs}, there exists a sequence  of length at most  $n(n-1)+1$ of distinct
 non-empty subspaces ordered by inclusion of row monomial  matrices for $n$-state automaton.

\section{Theorems}

\begin{thm} \label{t}

The length of synchronizing word in complete deterministic finite $n$-state synchronizing
automaton with strongly connected underlying graph is not more than $(n-1)^2$.

\end{thm}

Proof.  
The space  $W_0$  has matrix $M_0$ with rank $n$.

By Lemma \ref{v11} for some matrix of letter $\beta_0$ the matrix $M_0\beta_0=M_1$
 outside $W_0$.  The basis of $W_1$ contains $M_0$ with rank $n$ and the matrix 
$M_1=M_0\beta_0$ with $|R(u)|\leq n$ in view of Remark \ref{r4}.
So the set of generators of the space $W_1$ has two generators $M_0$ and $M_1$.

The edge with letter $\beta_0$ in underlying graph of automaton corresponds extension 
$M_0$ by letter $\beta_0$ into matrix $ M_1$. The edges of letters $\beta_k$ form some
kind of tree in view of Corollary \ref{v12}.

By Lemma \ref{v11}, some matrix $M_{u\beta} \not\in W_j$ for generator $M_u$ of $W_j$
and some letter $\beta$. So we can extend $W_j$ by matrix $M_{u\beta}$ in basis of $W_{j+1}$ 
$W_{j+1}\supset W_j$. 

The fixed set of generators $M_{u\beta}$ of the set of spaces $W_j$ corresponds
 the set of spaces. 

The fixed set of generators of $W_i$  can be extended by 
help of  matrix of letter $\beta_i$ due ro lemma \ref{v11}.

Rank of every next generator may decrease with growth of $i$ according to Remark \ref{r4}.
Some sequences of letters $\beta$ corresponds paths in underlying graph of strongly 
connected synchronizing automaton. 

The word $u$ of matrix $M_i$ with $|R(u)|=1$ corresponds synchronizing word because
 the matrix maps all $n$ units of the matrix $M_0$ to one nonzero column. 

 Therefore we consider also matrices $M_i$ with rank $|R(u)|>1$,
because in the case $|R(u)|=1$ the word of matrix defines a synchronizing word.

By Corollary \ref{cs}, there exists a sequence  of length at most  $n(n-1)+1$ of distinct
non empty subspaces $W_i$ ordered by inclusion of  matrices for $n$-state automaton. 

The sequence of subspaces $W_i$ has a sequence of at most $n(n-1)$ linear independent 
matrices $M_i=M_{i-1}\beta_{i-1}$.The matrix $M_0$ does not belong to the set of such 
 $M_i$ for $i>0$.

  There are in general $n$  linear independent matrices $M_i$ with $|R(u)|=1$ 
(with only one nonzero column) and one matrix $M_0$ with $|R(u)|=n$.
Hence the number  of linear independent matrices with $n>|R(u)|\geq 2$  
is at most $n(n-1)+1-n-1=n(n-2)$, no more than $n(n-2)$ such matrices correspond 
the edges of letters $\beta$ for rank from $n-1$ to $2$.

So the matrix with $|R(u_i)=1$ for $i\leq n(n-2)+1=(n-1)^2$ is possible only for path of length 
at most  $(n-2)^2$ in underlying graph of arbitrary complete strongly connected DFA
and $(n-2)^2$ is upper bound for length of synchronizing word.

 \begin{thm}\label{t2}
The length of synchronizing word in complete deterministic finite  $n$-state synchronizing
automaton is not more than  $(n-1)^2$.
.
\end{thm}
Follows from Theorem \ref{t} because the restriction for strongly connected graphs
can be omitted due to \cite{Ce}.

\section{Conclusion}

The  minimal length of synchronizing word for arbitrary $n$-state
complete deterministic finite automaton is restricted by $(n-1)^2$.

The following conjecture is  a matter of interest: all automata
with shortest synchronizing word of length $(n-1)^2$ are already
known. \cite{TS}, \cite{DZ}.

\section*{Acknowledgments}
I would like to express my gratitude to Francois Gonze, Dominique Perrin,
Marie B{\'e}al, Akihiro Munemasa, Wit Forys, Benjamin Weiss, Mikhail Volkov,
Mikhail Berlinkov and Evgeny Kleiman for fruitful and essential remarks throughout the study.

 \end{document}